\documentclass[12pt]{article}
\usepackage{graphicx}

\usepackage{amsmath}
\usepackage{amsthm}
\usepackage{amsfonts}
\usepackage{amssymb}
\usepackage[english]{babel}

\title{Geobiodynamics and \\Roegenian Economic Systems\footnote{improved version of the paper \cite{[18]}}}
\author{C. Udriste, M. Ferrara, \\D. Zugravescu, F. Munteanu, I. Tevy}
\date{}

\begin{document}
\maketitle
\textheight 19cm
\textwidth 16cm
%\oddsidemargin 2,5cm
%\evensidemargin 2,5cm
\topmargin 0,1cm
\newcommand{\di}{\displaystyle}
\newcommand{\ov}{\over}
\newcommand{\noa}{\noalign{\medskip}}
\newcommand{\al}{\alpha}
\newcommand{\be}{\beta}
\newcommand{\om}{\omega}
\newcommand{\bt}{\bar\tau}
\newcommand{\br}{\hbox{\bf R}}
\newcommand{\bd}{\hbox{\bf D}}
\newcommand{\bp}{\hbox{\bf P}}
\newcommand{\bn}{\hbox{\bf N}}
\newcommand{\bee}{\hbox{\bf E}}
\newcommand{\bc}{\hbox{\bf C}}
\newcommand{\bb}{\hbox{\bf B}}
\newcommand{\bv}{\hbox{\bf V}}
\newcommand{\sgn}{\hbox{sgn}}
\newcommand{\rot}{\hbox{rot}}
\newcommand{\divv}{\hbox{div}}
\newcommand{\bh}{\hbox{\bf H}}
\newcommand{\bm}{\hbox{\bf M}}
\newcommand{\ome}{\Omega}
\newcommand{\ba}{\bar a}
\newcommand{\ti}{\times}
\newcommand{\ga}{\gamma}
\newcommand{\ty}{\infty}
\newcommand{\de}{\delta}
\newcommand{\te}{\theta}
\newcommand{\na}{\nabla}
\newcommand{\pa}{\partial}
\newcommand{\va}{\varphi}
\newcommand{\ld}{\ldots}
\newcommand{\qu}{\quad}
\newcommand{\la}{\lambda}
\newcommand{\fo}{\forall}
\newcommand{\ep}{\varepsilon}
\newcommand{\pp}{\prime}
\newcommand{\su}{\subset}
\newcommand{\si}{\sigma}
\newcommand{\dd}{\Delta}
\newcommand{\gaa}{\Gamma}
\newcommand{\ri}{\Rightarrow}
\newcommand{\rii}{\Leftrightarrow}
\newcommand{\med}{\medskip}

\theoremstyle{definition}
\newtheorem{definition}{Definition}[section]
\theoremstyle{theorem}
\newtheorem{theorem}{Theorem}[section]
\newtheorem{proposition}{Proposition}[section]
\newtheorem{corollary}[theorem]{Corollary}
\newtheorem{lemma}[theorem]{Lemma}
\newtheorem{remark}{Remark}[section]

\begin{abstract}
This mathematical essay brings together ideas from Economics, Geobiodynamics and Thermodynamics. Its purpose is to obtain real models of complex evolutionary systems. More specifically, the essay defines Roegenian Economy and links Geobiodynamics and Roegenian Economy. In this context, we discuss the isomorphism between the concepts and techniques of Thermodynamics and Economics. Then we describe a Roegenian economic system like a Carnot group. After we analyse the phase equilibrium for two heterogeneous economic systems. The European Union Economics appears like Cartesian product of Roegenian economic
systems and its Balance is analysed in details.
The Phase Diagram for an economic system highlights a triple point.
A Section at the end describes the "economic black holes" as small parts of a
a global economic system in which national income is so great that it causes others poor enrichment.
These ideas can be used to improve our knowledge and understanding of the nature of development and evolution of thermodynamic-economic systems.

\end{abstract}

{\bf AMS Mathematical Classification}: 80A99, 91B74, 83C57

{\bf J.E.L. Classification: B41}

\footnote{it is started on the ideas
of the mathematician and economist Nicholas Georgescu Roegen},
Thermodynamic-economic dictionary, Carnot group, sub-Riemannian structure, phase equilibrium,
balance of European Union economy, economic equilibrium, economic 3D black holes.

\section{Some reflections on artificial-natural \\dualism}

Like any living system (or fundamental system for life), the Earth System understood like the Planet,
is extremely difficult to define and characterize coherently and completely.
At present, the society is found in a modern age of its evolution in which the informational society is
rising. At the same time, a society based on
knowledge (which is not yet fully mature or functional) can be regarded as being in an emergent state. Given this general context, we will try to define a Planet-Earth System that extends beyond its natural (inorganic, organic or life) structures, which are usually the subject of study for Geosciences and/or Biology and Biochemistry. However, besides natural structures, man-made artifacts have also been produced throughout the history of mankind and, after adding and combining with natural ones, have led to a complex symbiosis (or to a new concept of the Economy known as "the whole common living"). Currently, a possible list of such man-made systems with significant (ecological, economic, political, cultural and social) impact may include elements such as:

(a) the whole set of systems intended for the generation and distribution of electricity. To these we add oil extraction and refining as well as distribution of petroleum products to industrial and individual customers;

b) the entire combined road, rail and airline network that make up the necessary infrastructure for commerce and tourism, involving the movement of goods and people;

c) production and distribution of activities of any type of goods (merchandise or non-cargo).

All these systems created by man (see Popescu-Tasnadi book [8]),
when interacting with the natural ones in which they are found and / or
with which they interact for their normal coexistence, generates completely new systems of a different qualitative type. The natural and artificial systems that are now being mixed, interacting, interconnected and interdependent, permanently defining and influencing each other, giving birth to a new and qualitatively different dynamics. This new symbiosis and its dynamics are only partially controllable or predictable, and to be understood (even only partially), they need to be examined through an interdisciplinary approach, a re-evaluation of classical sciences, a fundamental change in the views and techniques used to study, model and predict such complex systems [Holden 1996].

Investigations have been initiated and research is taking place in the wider world
to elucidate which are the best study directions for this artificial-natural hybrid. From the multitude of the results obtained so far as a result of these activities, it remains to be mentioned that "the only method" capable of studying an object of such holism and diversity, with an integrative approach, is Complexity Science [Internet 2006]. In this new vision, the behavior and dynamics of the Planet-Earth system are simultaneously determined by two major force systems:

a) complicated interactions between crust, atmosphere, hydrosphere and ionosphere,
whose activity is further influenced and modulated by system dynamics
solar (electromagnetic storms, solar wind, ionospheric and tellurian currents,
sunspots and eruptions, tides etc.);

b) essential social components that define new needs and opportunities in their evolution, resulting in a constantly changing world of artifacts, and therefore a permanent coupling that alters the interaction between natural and artificial parts of the planet.

From this perspective, studies / sciences, sociological, economic or engineering
must be reconfigured and integrated into a larger and broader subject (meta-science) that transcends and interdisciplinary recombination in order to create a new framework in which each part depends and is supported by the other elements.
In the first step, the generation and application in practice of the aforementioned phenomena
has already begun and a growing number of researchers refer more frequently
to such new meta-domains as bioeconomics, biogeo-geology, geobiophysics,
astrobiophysics as bioelectronics, microelectromechanics and jurisdicamics.
Although they are not yet fully mature or have not fully developed as meta-sciences,
we can still consider them as intermediary stages in the journal of the integration of several different sciences into a new current field. We can appreciate that this can be identified as the first sign of integration of the various disciplines, which eventually leads to a single global concept, probably similar in a certain
measure with the Gaia concept [Lovelock 1987, Lenton 1998]. It should be clear that once such a meta-science is generated, disciplines such as economics and sociology will no longer be studied separately or independently but interdependent and always in their context of interactive co-evolution with the Planet system -Earth.

Therefore, in this new context, increased interest and more intense studies can be expected
in the following possible directions:

(1) stimulating the transfer of knowledge between different areas and encouraging pluri- and inter-disciplinary approaches;

(2) assessing the capabilities of today's methodologies to understand theoretically
and experimental transition from one part to the whole, from complicated to complex;

(3) the discovery or invention of new experimental concepts, models, theories, methods and techniques of monitoring and evaluation of hierarchical dissipative systems, evolving away from the thermodynamic equilibrium;

(4) the successful development and use of educational infrastructure that can provide transfer and filtering of information and bring specific knowledge and know-how.

The main objective is to educate the new generation and re-educate the generation
current changing the current reductionist paradigm to a new one of nonlinearity and complexity. This should lead to a better understanding of current phenomena, capacity building and the desire to assimilate new knowledge and to adopt a framework for exploring the mind in order to generate new knowledge. Therefore, a long-term consequence of a new educational infrastructure should be the creation and dissemination of a lifelong learning lifestyle based on formal education, but also including an informal one and a non-formal one -formal, as well as the inclusion of both localized and delocalized aspects (eg e-learning).

Nicholas Georgescu-Roegen (in Romanian, Nicolae Georgescu) - born: Constan\c ta, Romania, February $4, 1906$; died: Nashville, Tennessee, October $30, 1994$)- was a Romanian-American mathematician, statistician and economist best known for his Magnum Opus, $1971$, which considers that the second law of thermodynamics (free energy tends to disperse or lose itself in the form of bent energy) governs economic processes as well, introducing a new domain called thermoeconomics.

Roegen introduced into the economy the concept of entropy, similar to that in thermodynamics, creating the foundation that later developed in the evolutionary economy. His research has also contributed significantly to the development of bio-economy and eco-economy.

In this paper we will develop the vision of Nicholas Georgescu Roegen, which links the economic phenomena of entropy. For this purpose, we will identify a formal correspondence between economic processes and thermodynamic processes, with heuristic implications for the characterization of the dynamics and evolution of economic systems. Then, on the basis of this formalism, we will study the problem of establishing economic equilibrium conditions in the case of aggregation of initially stable and independent economic subsystems in a functional entity, e.g. the architecture of the European Community Economy. Similarly, we intend to expand the concept
of the "black hole" from astrophysics to the economy. This is possible by applying mathematical modeling to economic phenomena the formalism commonly used to model entropic processes. In order to successfully apply such formalism we must create a dictionary that reflects the thermodynamics-economy dualism,
identifying correspondences between physical and economic parameters. This dictionary
it is an instrument and also an example given in this essay that illustrates
the ability to transfer and use knowledge and theories from an initial domain,
well established and structured (Thermodynamics, in this case), for understanding and modeling in a controversial field (Macro-economy).

\section{Thermodynamic-Economic Dictionary}

Thermodynamics is important as a model of phenomenological theory, which describes and
unifies certain properties of different types of physical systems. There are many
systems in biology, economics and computer science, for which an organization
similar and unitary-phenomenological would be desirable. Our purpose is to present
certain features of the economy that are inspired by thermodynamics and vice versa. In this context, we offer a Dictionary that reflects the Thermodynamics-Economy isomorphism. The formal analytical-mathematical analogy between economics and thermodynamics is now well-known or at least accepted by economists and physicists as well. [See also Econophysics which is a discipline in this sense].

Starting from these observations, the works of Udriste et al. (2002-2013) build an isomorphism between thermodynamics and the economy, admitting that fundamental laws are also in correspondence through our identification. Therefore, each thermodynamic system is naturally equivalent to an economic system, and thermodynamic laws have correspondence in the economy.

In the following, we reproduce the correspondence between the characteristic state variables and the laws of thermodynamics with the macro-economics as described in Udri\c ste et al. (2002-2013) based on the theory on which the Roegenian economy is founded in 1971
\cite{[3]}. They also allow us to include in the economy the idea of a "black hole" with a similar meaning to the one in astrophysics [Udri\c ste, Ferrara 2008]. We do not know if Roegen would have judged this, but that's why we are judging him instead.

\vspace{1cm}
{\hspace{-1cm}
\begin{tabular}{lll}
\vspace{0.3cm}
\hspace{0.5cm}THERMODYNAMICS&\hfill & \hspace{0.7cm}ECONOMICS \\
U=\hbox{internal energy} & \hfill\ldots &G=\hbox{growth potential}\\
T=\hbox{temperature}&\hfill\ldots& I=\hbox{internal politics stability}\\
S=\hbox{entropy}&\hfill \ldots& E=\hbox{entropy}\\
P=\hbox{pressure}& \hfill \ldots& P=\hbox{price level (inflation)}\\
V=\hbox{volume}&\hfill \ldots& Q=\hbox{volume, structure, quality}\\
M=\hbox{total energy (mass)}&\hfill \ldots& Y=\hbox{national income (income)}\\
Q=\hbox{electric charge}& \hfill \ldots& $\mathcal{I}$=\hbox{total investment}\\
J= \hbox{angular momentum}&\hfill \ldots& J=\hbox{economic angular momentum}\\
\hspace{1cm}\hbox{(spin)}&\hfill &\hspace{1cm}\hbox{(economic spin)}\\
M=M(S,Q,J)& \hfill \ldots& Y=Y(E,{$\mathcal{I}$},J)\\
$\Omega = \frac{\partial M}{\partial J}$= \hbox{angular speed}&\hfill \ldots& $\frac{\partial Y}{\partial J}$=\hbox{marginal inclination to rotate}\\
$\Phi = \frac{\partial M}{\partial Q}$=\hbox{electric potential}&\hfill \ldots& $\frac{\partial Y}{\partial {\cal{I}}}$=\hbox{marginal inclination to investment}\\
$T_H= \frac{\partial M}{\partial S}$=\hbox{Hawking temperature}&\hfill \ldots& $\frac{\partial Y}{\partial E}$=\hbox{marginal inclination to entropy}\\
$G$ = Newton constant &\hfill \ldots& $\mathcal{G}$= universal economic constant\\
$c$ = light velocity &\hfill \ldots& $c$= maximum universal exchange speed \\
$\hbar$= normalized Planck constant &\hfill \ldots& $\hbar$ = normalized economic quantum \\
\end{tabular}}
\vspace{1cm}

{\it The Gibbs-Pfaff fundamental equation in thermodynamics} $ dU-TdS + PdV + \sum_k \mu_k dN_k = 0 $ is changed to
{\ Gibbs-Pfaff fundamental equation of economy} $ dG-IdE + PdQ + \sum_k \nu_k d {\cal {N}} _ k = 0 $. These equations are combinations of the first law and the second law (in thermodynamics and economy respectively). The third law of thermodynamics $ \di \lim_{T \to 0} S = 0 $ suggests the third law of economy $ \di \lim_ {I \to 0} E = 0 $ "if the internal political stability $ I $ tends to $ 0 $, the system is blocked, meaning entropy becomes $ E = 0 $, equivalent to maintaining the functionality of the economic system must cause disruption").

Process variables $ W = $ {\it works mechanically} and $ Q $ = {\it heat} are introduced into elemental mechanical thermodynamics by $ dW = PdV $ (the first law) and by elementary heat, respectively, $ dQ = TdS $, for reversible processes,
or $ dQ <TdS $, for irreversible processes ({\it second law}). Their correspondence in the economy, $$ W = \hbox {\it the wealth of the system},\,\,q= \hbox{\it production of goods}$$
are defined by $ dW = Pdq $ ({\it elementary wealth in the economy}) and $ dq = IdE $ or $ dq <IdE $ ({\it the second law or the elementary production of commodities}\footnote{a commodity
is an economic good, a product of human labor, with a utility in the sense of life, for sale-purchase on the market in the economy}).

Sometimes a thermodynamic system is found in an
{\it external electromagnetic field} $ (\vec {E}, \vec {H}) $. {\it The external electric field} $ \vec E $ determines
the {\it polarization} $ \vec P $ and {\ it the external magnetic field} $ \vec H $ determine {\it magnetizing} $ \vec M $. Together they give the total elementary mechanical work  $ dW = PdV + \vec E d \vec P + \vec H d \vec M $. Naturally, an economic system is found in an {\it external econo-electromagnetic field} $(\vec {e}, \vec {h})$. The {\it external investment (econo-electric) field} $\vec e$ determines
{\it initial growth condition field (econo-polarization field)} $ \vec p $
and the {\it external growth field (econo-magnetic field)} $ \vec h $
cause {\it growth (econo-magnetization)} $ \vec m $. All these fields produce the elementary mechanical work $dW=PdQ+\vec e
d\vec e +\vec h d\vec m$. The economic fields introduced here are imposed on the one hand by the type of economic system and on the other hand by the policy makers (government, public companies, private firms, etc.).

The long term association between Economy and Thermodynamics can be strengthened
with new tools based on the previous dictionary. Of course, this new idea of
the thermodynamically-economical dictionary produces concepts different from those in econophysics. Econophysics seems to build similar economic notions to physics
(Cernavsky et al 2002, London, Tohme 2007, Ruth 2005), as if those in the economy were not enough. The thermodynamic-economic dictionary allows the transfer of information from one discipline to another (Udri\c ste et al., 2002-2013), keeping the background of each discipline, that we think that was suggested by Roegen in 1971.

A macro-economic system based on a Gibbs-Pfaff equation is controllable (see \cite{[19]}).
\begin{definition}
{\it An economy structured similar to thermodynamics is called Roegenian economy}.
\end{definition}
Economics described by Gibbs-Pfaff equation is Roegenian economics.

Roegenian economics (also called bioeconomics by Georgescu-Roegen) is both a
transdisciplinary and an interdisciplinary field of academic research addressing
the interdependence and coevolution of human economies and natural ecosystems,
both intertemporally and spatially. By treating the economy as a subsystem of
Earth's larger ecosystem, and by emphasizing the preservation of natural capital,
the field of ecological economics is differentiated from environmental economics,
which is the mainstream economic analysis of the environment.

\section{Roegenian economic system like \\Carnot group}

A Roegen economic system benefits of two specific composition laws on points in the manifold $\mathbb{R}^5$:
one commutative, the other non-commutative. They show that compounding by addition is far from economic truth.

{\bf Commutative group}

On a closed Roegenian economic system  $(\mathbb{R}^5,\,\, dG-IdE+PdQ=0)$ it is possible to introduce a commutative group structure
using a square operation on $ \mathbb{R}^5 $, namely
$$(G_1,I_1,E_1,P_1,Q_1)\bullet (G_2,I_2,E_2,P_2,Q_2)$$
$$= (G_1+G_2-I_1E_2 + P_1Q_2-I_2E_1+P_2Q_1, I_1+I_2,E_1+E_2,P_1+P_2,Q_1+Q_2).$$
Commutative group axioms are verified by direct calculations.

{\bf Non-commutative Carnot group}

{\bf Algebraic variant}

On a closed Roegenian economic system $(\mathbb{R}^5, \, dG-IdE + PdQ = 0) $, a Carnot group structure can also be introduced, using another square operation on $ \mathbb{R}^5 $, namely
$$(G_1,I_1,E_1,P_1,Q_1)\star (G_2,I_2,E_2,P_2,Q_2)$$
$$= (G_1+G_2+ \frac{1}{2}\,AC, I_1+I_2,E_1+E_2,P_1+P_2,Q_1+Q_2),$$
where
$$AC = (I_1E_2 -I_2E_1)+ (P_1Q_2 -P_2Q_1) + (E_1G_2-E_2G_1).$$
Non-commutative group axioms are verified by direct calculations.

{\bf Geometric variant}

The frame of Roegen distribution  $dG-IdE+PdQ=0$ is made up of vector fields
$$X_1=\frac{\partial}{\partial I},\,\,X_2=I\frac{\partial}{\partial G}+\frac{\partial}{\partial E},
\,\,X_3=\frac{\partial}{\partial P},\,\,X_4=\frac{\partial}{\partial Q}-P\frac{\partial}{\partial G}.$$
The five-dimensional Roegen algebra is a Lie algebra with the basis $X_1,X_2,X_3,X_4$
together with Poisson parenthesis relationships
$$[X_1,X_2]=X_5,\,[X_1,X_3]=0,\, [X_1,X_4]=0$$
$$[X_2,X_3]=0,\,[X_2,X_4]=0,\,[X_3,X_4]=-X_5.$$

The Roegen Lie Algebra is nilpotent, i.e.,
$$[X_1,[X_1,X_2]]=0,\,[X_2,[X_1,X_2]]=0,\, [X_3,[X_1,X_2]]=0,\, [X_4,[X_1,X_2]]=0.$$
On the other hand, any Lie nilpotent algebra has a single simple Lie group, called the nilpotent group. In our case, we have a  nilpotent group of step $2$ isomorphic with the previous algebraic group.

\section{Sub-Riemannian structure on \\Roegenian economy}

Let us consider the Roegenian economy (an economic distribution) generated by the Gibbs-Pfaff fundamental
equation $dG-IdE + PdQ =0$ on $\mathbb{R}^5$. The normal vector field to this distribution,
in $\mathbb{R}^5$, is $N=(1,0,-I,0,P)$.
Being a contact form, the integral manifolds are only curves (that are not uniquely determined by
initial condition (point, tangent vector)) and surfaces (that are not uniquely determined by
initial condition (point, tangent plane)).

The frame of Roegen distribution $dG-IdE+PdQ=0$ is made up of vector fields
$$X_1=\frac{\partial}{\partial I},\,\,X_2=I\frac{\partial}{\partial G}+\frac{\partial}{\partial E},
\,\,X_3=\frac{\partial}{\partial P},\,\,X_4=\frac{\partial}{\partial Q}-P\frac{\partial}{\partial G}.$$
This means $<N,X_a>=0,\,\,a=1,2,3,4$.
The dual frame, defined by $\omega^a(X_b)= \delta^a_b,\,\, a,b =1,2,3,4$, is
$$\omega^1= dI,\,\, \omega^2= dE,\,\, \omega^3= dP,\,\, \omega^4= -\frac{1}{P}\,dG+\frac{I}{P}\,dE,$$
on $\mathbb{R}^5\cap \{(G,I,E,P,Q)\,\vert \, P\neq 0\}$.
The sub-Riemannian metric is given by the square of arc-element
$$ds^2 = \delta_{ab}\omega^a \omega^b= \delta^{ab}X_aX_b$$
or explicitly
$$ds^2=\frac{1}{P^2}\,dG^2+dI^2+\left(1+\frac{I^2}{P^2}\right)dE^2 - \frac{2I}{P^2}\,dG\, dE + dP^2,$$
on $\mathbb{R}^5\cap \{(G,I,E,P,Q)\,\vert \, P \neq 0\}$.
The components of this sub-Riemannian metric are rational functions.
\begin{theorem} The non-zero components of the Christoffel symbols of the second kind are (rational functions)
$$\Gamma^1_{12}=-\frac{I}{2P^2},\,\,\Gamma^1_{14} =- \frac{1}{P},\,\,\Gamma^1_{23} =- \frac{P^2-I^2}{2P^2},\,\,\Gamma^1_{34}=\frac{I}{P}$$
$$\Gamma^2_{13} = \frac{1}{2P^2},\,\,\Gamma^2_{33} =- \frac{I}{P^2},\,\,
\Gamma^3_{12}=-\frac{1}{2P^2}$$
$$\Gamma^3_{23} = \frac{I}{2P^2},\,\,\Gamma^4_{11} = \frac{1}{P^3},
\,\,\Gamma^4_{13}=-\frac{I}{P^3},\,\,\Gamma^4_{33} = \frac{I^2}{P^3},$$
on $\mathbb{R}^5\cap \{(G,I,E,P,Q)\,\vert \, P \neq 0\}$.
\end{theorem}
\begin{theorem}
The equations of geodesics are
$$\ddot I+\frac{1}{P^2}\,\,\dot G \,\dot E - \frac{I}{P^2}\,\,\dot E\, \dot E=0$$
$$\ddot E-\frac{1}{P^2}\,\,\dot G \,\dot I - \frac{I}{P^2}\,\,\dot E\, \dot I =0$$
$$\ddot P+\frac{1}{P^3}\,\,\dot G \,\dot G - \frac{2I}{P^3}\,\,\dot G\, \dot E +\frac{I^2}{P^3}\,\,\dot E\,\dot E=0$$
$$\ddot G-\frac{I}{P^2}\,\,\dot G \,\dot I - \frac{2}{P}\,\,\dot G\, \dot P -\,\,\frac{P^2-I^2}{P^2}\,\,\dot I\,\dot E
+ \,\, \frac{2I}{P}\,\,\dot E\, \dot P=0,$$
on $\mathbb{R}^5\cap \{(G,I,E,P,Q)\,\vert \, P \neq 0\}$.
\end{theorem}

If $t\to \gamma(t)=(G(t),I(t),E(t),P(t),Q(t))$
is a geodesic curve, then
$${\displaystyle G(t)=\int _{\gamma}IdE - PdQ}$$
with the integral limited to a four-dimensional hyperplane.

\begin{remark}
All the geometry of the economic distribution is characterized by
rational functions since the components of sub-Riemannian metric are rational functions.
\end{remark}

\section{Phase equilibrium for \\heterogeneous economic systems}

Let us transpose into the economy a thermodynamic model shown in Cre\c tu's book
\cite{[2]} pp. 149-150.

Heterogeneous economic systems are comprised of two or more homogeneous areas from a macroeconomic point of view. The homogeneous domains of an economic system are called phases. An economic system is understood to be an indivisible economic entity. For example, an economic system containing commodities and money consists of two phases and a single "value" component.

As a working hypothesis, we assume an economic system made up of one
component (company) and two phases (commodity and money) and introduce the notion of {\it economic mole} as a unit of measure of value (value merchandise, money value).
Let $ m_1 $ and $ m_2 $ be the number of {\it economical moles} in each phase, $ g $ (growth potential), $ q $ (production unit), $ e $ (entropy), reported at an economic mole. The economic system under study is considered to be isolated. That's why we have
$$G = G_1 + G_2 = m_1g_1+m_2g_2 = \hbox{const}$$
$$Q = Q_1 + Q_2 = m_1q_1+m_2q_2 = \hbox{const}$$
$$m=m_1+m_2 = \hbox{const}.$$
To reach equilibrium points, we need the equations
attached to the differentials of the previous functions,
$$dm_1 + dm_2 =0$$
$$g_1dm_1 + m_1dg_1 + g_2dm_2 + m_2dg_2=0$$
$$q_1dm_1 + m_1dq_1 + q_2dm_2 + m_2dq_2=0$$
or
$$dm_2 =- dm_1$$
$$m_2dg_2= -g_1dm_1 - m_1dg_1 - g_2dm_2= -g_1dm_1 -m_1dg_1+g_2dm_1 $$
$$m_2dq_2= -q_1dm_1 - m_1dq_1 - q_2dm_2 = -q_1dm_1 - m_1dq_1 + q_2dm_1.$$
On the other hand, at economic equilibrium, the entropy of the system is maximum,
$$E= E_1+E_2 = m_1 e_1 + m_2 e_2 = E_{max}$$
and therefore the free critical point condition is written
$$dE= e_1dm_1+ m_1de_1 + e_2dm_2 + m_2de_2 =0.$$
Add the initial links (in the differential form) and the links (equations) Gibbs-Pfaff
$$dg_1=I_1 de_1 - P_1 dq_1, dg_2=I_2 de_2 - P_2 dq_2,$$
to get critical point condition with restriction.
Replacing in $dE=0$, we find the identity
$$dE = m_1\left(\frac{1}{I_1}-\frac{1}{I_2}\right)dg_1 + m_1\left(\frac{P_1}{I_1}-\frac{P_2}{I_2}\right)dq_1$$
$$+\left(\left(e_1-\frac{g_1+ P_2q_1}{I_2}\right) - \left(e_2-\frac{g_2+ P_2q_2}{I_2}\right)\right)dm_1=0.$$
Since the differentials $dg_1, dq_1, dm_1$, within this identity, are arbitrary,
the economic equilibrium condition is obtained
identifying the coefficients with zero, hence
$$I_1 = I_2 = I,\,\, P_1=P_2 = P,\,\, g_1 + P_1q_1 - Ie_1 = g_2 + P_2q_2 - Ie_2.$$
In other words, the balance is on equal domestic policies and equal prices.
Moreover, there is an economic quantity $\mu = g+ Pq -I e$,
which has the same value
for the two phases in economic equilibrium. The economical quantity $\mu$ is
called economic potential of Gibbs type, relative to an economic mole (unit of measure for values).

\section{The Balance of European Union Economy}

The reasoning in Section $ 3 $ can also be used for the problems in this Section. But, to take into account the theory of nonholonomic constraints and Lagrange multiplier theory, specific to optimization with constraint, we prefer the method applied by Udri\c ste et al. (2002-2013) in their research, which is more rigorous in such situations.

The European Union (EU) is a union of $27$ independent states based
on the European Community and founded to increase political, economic and social cooperation. Therefore, we need to analyze
equilibrium states after interaction of $ 27 $ simple economic systems
(each with five independent variables, $G, I, E, P, Q$):
$$\mathbb{R}^5 (\hbox{space}),\,\,\omega_1=dG_1 - I_1dE_1 +P_1dQ_1=0 \,(\hbox{Pfaff equation}),$$
$$\mathbb{R}^5 (\hbox{space}),\,\,\omega_2=dG_2 - I_2dE_2 +P_2dQ_2=0 \,(\hbox{Pfaff equation}),$$
$$\ldots \hspace{2cm}\ldots\hspace{2cm}\ldots$$
$$\mathbb{R}^5 (\hbox{space}),\,\,\omega_{27}=dG_{27} - I_{27}dE_{27} +P_{27}dQ_{27}=0\, (\hbox{Pfaff equation}).$$

The evolution of each simple economic system can only be a $ 1 $ -dimensional or $ 2$-dimensional manifold, because each $ 1 $ -form
$ \omega_i, i = 1, ..., 27, $ is a contact form. The Gibbs-Pfaff equation $ dG-I \, dE + P \, dQ = 0 $ generates a distribution on the space $ \mathbb{R}^5 $ in the sense that through the normal vector field $$ Z (G, I, E, P, Q) = (1,0, -I, 0, P) $$ attached to
the distribution, to each point $ M (x_{i_1}, ..., x_{i_5}) $ of $ \mathbb{R}^5 $ one attaches a hyperplan (four dimensions) with the normal vector field $ Z(M) $.

The best way to analyze the interaction is to consider the economic system of Cartesian  product type
$$\mathbb{R}^{135}=R^{5\times 27},\,\,\omega_i=0,\,\,i=1,...,27$$
and find the constrained critical points of some aggregate objective functions.

\subsection{Steady-state equilibrium}

Idealized theories are good for making predictions.

Often, in the classical economic literature, the usual objective functions are
$$G=\sum_{i=1}^{27}G_i\,\,\,\hbox{total increase}$$
$$E=\sum_{i=1}^{27}E_i\,\,\,\hbox{total entropy}$$
$$Q=\sum_{i=1}^{27}Q_i\,\,\,\hbox{total quantity of products}.$$

An equilibrium should be described by the critical points of one of the three objective functions constrained by the constant levels of the other two functions and by $ 28 $ Gibbs-Pfaff equations  $\omega_i = 0, \, i = 1, ..., 27$. This kind of economic research has been in our attention since 2001 (see Udri\c ste et al. 2002-2013).

\begin{theorem}
The critical points of total increase $G$ constrained by
$$E=\hbox{const.},\,\, Q=\hbox{const}.,\,\, \omega_i=0,\,\,i=1,...,27$$ are the points of the nonholonomic manifold
$$\mathbb{R}^{135}=R^{5\times 27},\,\,\omega_i=0,\,\,i=1,...,27$$ at which
 $$I_1=...=I_{27}\,\, \hbox{and}\,\, P_1=...=P_{27}$$
(the same stable governance policies $I$, the same prices $P$).
\end{theorem}
\begin{proof}
The critical points of the aggregate objective function named "total increasing"
 $G=\sum_{i=1}^{27}G_i$ subject to the retrictions
$$E=\sum_{i=1}^{27}E_i=\hbox{const.},\,\,Q=\sum_{i=1}^{27}Q_i=\hbox{const.},\,\,\omega_i=0,\,\,i=1,...,27$$
are zeros of {\it Lagrange $1$-form}
$$\sum_{i=1}^{27}dG_i + \sum_{i=1}^{27}\lambda^i\omega_i+\lambda^{28}\sum_{i=1}^{27}dE_i +\lambda^{29}\sum_{i=1}^{27}dQ_i.$$
It follows
$$\lambda^i+1=0,\,\,-\lambda^iI_i+\lambda^{28}=0,\,\,\lambda^iP_i+\lambda^{29}=0,$$
and we find
$$\lambda^i=-1,\,\,\lambda^{28}=-I_1=...=-I_{27},\,\,\lambda^{29}=P_1=...=P_{27}.$$
In conclusion, at equilibrium we must have
$$I_1=...=I_{27},\,\,P_1=...=P_{27}.$$
\end{proof}
Similar arguments prove the following two theorems.
\begin{theorem}
The critical points of total entropy $E$ (aggregate objective function)
subject to the restrictions
$$G=\hbox{const}.,\,\,Q=\hbox{const}.,\,\,\omega_i=0,\,\,i=1,...,27$$
are the points of the nonholonomic manifold $$R^{135}=R^{5\times 27},\,\,\omega_i=0,\,\,i=1,...,27$$
at which $$I_1=...=I_{27}\,\, \hbox{and}\,\, P_1=...=P_{27}.$$
\end{theorem}

\begin{theorem} The critical points of the total production $Q$ (aggregate objective function)
subject to the restrictions
$$G=\hbox{const}.,\,\,  E=\hbox{const}.,\,\,  \omega_i=0,\,\,i=1,...,27$$
are the points of the nonholonomic manifold $$R^{135}=R^{5\times 27},\,\,\omega_i=0,\,\,i=1,...,27$$
at which $$I_1=...=I_{27}\,\, \hbox{and}\,\, P_1=...=P_{27}.$$
\end{theorem}
This utopian balance is the same for all objectives, total growth, total entropy,
the total quantity of products previously set.

Geometric, the set of all critical points is a hyperplane in Cartesian product $ R^{135} = R^{5 \times 27}$ space.

Economically, the utopian balance of interconnected economic systems can only be
achieved for "equal domestic political stability" and "relatively equal price levels
(inflation)". Being an utopian idea, it rests as a simplistic convergence criteria.
Therefore, the imbalance in the European Community is evident from previous mathematical point of view.

On the other hand, our Gibbs-Pfaff models are sustainable. To adjust the balance we need
specific strategies, like, for example strategy of weighted averages.

\subsection{Adjusted balance}

We replace the total increase (total entropy, total production quantity) by the weighted averages
$$\mathbb{G}=\alpha^i G_i,\,\,\sum_{i=1}^{27}\alpha^i=1$$
$$\mathbb{E}=\beta^iE_i,\,\,\sum_{i=1}^{27}\beta^i=1$$
$$\mathbb{Q}=\gamma^i Q_i,\,\,\sum_{i=1}^{27}\gamma^i=1.$$
Then the set of critical points leads to a possible
harmony\footnote{harmony has attributes: adjustment, coherence and resonance}
for the European Community Economy.

\begin{theorem}
The critical points of total weighted increase $\mathbb{G}$ constrained by
$$\mathbb{E}=\hbox{const.},\,\, \mathbb{Q}=\hbox{const}.,\,\, \omega_i=0,\,\,i=1,...,27$$
are the points of the nonholonomic manifold
$$\mathbb{R}^{135}=R^{5\times 27},\,\,\omega_i=0,\,\,i=1,...,27$$ at which
 $$\mathbb{I}=\frac{\lambda^1}{\beta^1}\,I_1=...=\frac{\lambda^{27}}{\beta^{27}}\,I_{27}\,\,\, \hbox{and}\,\,\, \mathbb{P}=\frac{\lambda^1}{\gamma^1}\,P_1=...= \frac{\lambda^{27}}{\gamma^{27}}\,P_{27}$$
(the same stable weighted governance policies $\mathbb{I}$, the same weighted prices $\mathbb{P}$).
\end{theorem}
\begin{proof}
The zeros of {\it Lagrange $1$-form}
$$\sum_{i=1}^{27}\alpha^idG_i + \sum_{i=1}^{27}\lambda^i\omega_i+\lambda^{28}\sum_{i=1}^{27}\beta^idE_i +\lambda^{29}\sum_{i=1}^{27}\gamma^idQ_i$$
are
$$\lambda^i+\alpha^i=0,\,\,-\lambda^iI_i+\lambda^{28}\beta^i=0,\,\,\lambda^iP_i+\lambda^{29}\gamma^i=0,$$
and we find
$$\lambda^i=-\alpha^i,\,\,\lambda^{28}=\frac{\lambda^i}{\beta^i}\,I_i,\,\,\lambda^{29}=-\frac{\lambda^i}{\gamma^i}\,P_i,
\,\,\hbox{for each}\,\, i=1,...,27.$$
\end{proof}

In conclusion, political stability and price stability must be an objective of
any monetary policy and we need to fix $\alpha^i, \beta^i, \gamma^i$ by some political conditions.

\section{Economic 3D Black Holes}

In astrophysics, a 3D black hole is a region of space in which the gravitational
field is so strong that nothing can escape after falling beyond the horizon event.
The horizon event is a point outside
Black holes where the gravitational attraction becomes so high that escape speed
(the speed at which an object could escape from the gravitational field)
equals the speed of light. On the other hand, the theory of relativity shows that no object can exceed the speed of light and therefore nothing, even the electromagnetic radiation (for example light) is unable to escape as soon as it is a little too distant from the center.
Black holes can only be detected by interacting with matter outside the horizon event, for example "by drawing" into the gas from an orbiting star.
While the idea of an object with gravity strong enough to
prevent the light from escaping was proposed in the century of the 18th century, the black holes in the current sense are described by Einstein's theory of general relativity,
drawn in $1916$. This theory predicts that when a sufficiently large amount of mass is present in a region small enough in space, all
the universal lines in space are deformed inward toward the center
body, forcing all matter and radiation to fall inside. While
general relativity theory describes a black hole as a vacuum space region with
pointwise singularity (gravitational singularity) at the center, and a horizon event
at the outer edge, the description changes when the effects
quantum mechanics are taken into account. Research on
this topic indicates that black holes can slowly release a form of
thermal energy called Hawking radiation, rather than keeping matter trapped forever.
However, a correct and final description of the black holes, which requires a theory of quantum gravity, is still unknown.

\subsection{3D black holes models in astrophysics}

In our next step, we use the following state variables:
entropy $ S $, mass (energy) $ M $, electric charge $ Q $ and spin $ J $.
Specific links between these variables, describing the near horizon geometry, define black holes.
Here we need to specify some simplifying conditions (to obtain geometrized units):
$G =1$, $ c =1$, where $G$ is the gravitational
constant and $c$ is light velocity (but, of course, we lose some physical dimensions).

1) {\it Reissner-Nordstrom black hole}:
A Reissner-Nordstrom black hole is made of non-rotating material but is electrically charged.
The near horizon geometry is characterized  by the surface
$$\hbox{either}\,\,\,S=2M^2-Q^2+2M^2\sqrt{1-\di{Q^2\ov M^2}}\,\,\,\,\hbox{or}\,\,\,\,M=\di{1\ov 2}\sqrt S+\di{Q^2\ov {2 \sqrt S}},$$
the second being of posynomial type. It follows $M\geq Q$.

2) {\it Kerr black hole}:
This pattern is based on one of the functions
$$\hbox{either}\,\,\,S=2M^2+2M^2\sqrt{1-\di\frac{J^2}{M^4}}\,\,\,\,\hbox{or}\,\,\,\,M=\di\frac{1}{2}\sqrt{\di\frac{4J^2+S^2}{S}}.$$
It follows $M\geq \sqrt{J}$.
The extremal limit of the kerr black hole occurs when $\di\frac{J}{M^2}=\pm 1$.

3) {\it BTZ black hole}: It is given through one of the functions
$$\hbox{either}\,\,\,S=2\sqrt{\frac{M}{4+J^2}}\,\,\,\, \hbox{or}\,\,\,\, M=S^2+\di\frac{J^2}{4S^2},$$
the second being of posynomial type. It follows $M\geq J$.

\subsection{3D black holes models in economics}

We recall what economists commonly call "economic black holes":
a business activity or product on which large amounts of money are spent, but that does
not produce any income or other useful result \cite{[0]}.

Let us introduce a new concept of "economic 3D black hole" as
a small part of a global economic system where the total income created
is so strong that nothing can escape after falling beyond the horizon event (rising poverty).
Here the economic black holes are depicted with entropy $ E $, national income (income)
$ Y $, total investment $ \mathcal{I} $ and economic spin $ J $ \cite{[17]}.
The national income is so large that it attracts all the economic resources of its neighbors.
This is in fact the image through the previous dictionary of a thermodynamic black hole.

1) {\it RN-economic black hole}:
The near horizon geometry is characterized  by the surface
$$\hbox{either}\,\,\,E=2Y^2-{\mathcal{I}}^2+2Y^2\sqrt{1-\di{{\mathcal{I}}^2\ov Y^2}}\,\,\,\,\hbox{or}\,\,\,\,Y=\di{1\ov 2}\sqrt E+\di{{\mathcal{I}}^2\ov {2 \sqrt E}},$$
the function $Y= Y(E,\mathcal{I})$ being a posynomial.

2) {\it K-economic black hole}:
This economic model is described by one of the functions
$$\hbox{either}\,\,\,E=2Y^2+2Y^2\sqrt{1-\di\frac{J^2}{Y^4}}\,\,\,\,\hbox{or}\,\,\,\,Y=\di\frac{1}{2}\sqrt{\di\frac{4J^2+E^2}{E}}.$$
The extremal limit of the kerr economic black hole occurs when $\di\frac{J}{Y^2}=\pm 1$.

3) {\it BTZ-economic black hole}:
$$\hbox{either}\,\,\, E=2\sqrt{\frac{Y}{4+J^2}}\,\,\,\,\hbox{or}\,\,\,\,Y=E^2+\di\frac{J^2}{4E^2},$$
the function $Y=Y(E,J)$ being a posynomial.

{\bf Open Problem} What represents each of the economical inequalities $Y\geq \mathcal{I},\,\, Y\geq \sqrt{J},\,\, Y\geq J$?

\section{Conclusions}

The present mathematical essay argues that in the real economy there are laws similar to those in thermodynamics, going from the term economy to the term Roegenian economy. The three important laws of the Roegenian economy are:

(i) {\bf First Economic Law}: $ dG = IdE-Pdq $, where $ dG $ is the infinitesimal increase
of growth potential, $IdE$ is the infinitesimal flow of internal politics stability in the economic system,
$ Pdq $ is infinitesimal wealth produced by the economic system. The first law is actually the law
conserving the potential for growth $ G $ (potential GDP in macroeconomics).

(ii) {\bf Second Economic Law (Law of Entropy)}: $ dE \geq 0 $ in the hypothesis of an isolated economic system. In other words, the entropy of an isolated system never decreases.

(iii) {\bf The third economic law (the law of economic death)}: $ E = 0 $ when $ I = 0 $, i.e., the absolute zero of the stability of internal policy implies "entropy is zero".

The variables $G, E, q$ are all extensive quantities. The variables $I, P$ are intensive quantities. Each pair $(I, E)$, or $(P, q)$ is called a conjugate pair relative to the $G$ growth potential. The pairs of intensive variables $(I, P)$ can be regarded as a generalized "economic force". An imbalance of an intense variable causes a "flux" of the associated variable in a direction that counteracts the imbalance.

Although unusual, our economic models can explain the strange phenomena of macroeconomics (collapse, crisis, turmoil, bankruptcy, etc). Making interest of economists on mathematical encryption will surely lead to the detection of techniques to prevent and predict economic collapse at the level of multinational companies and the level of GDP in the states that are part of unions.

Extended Roegen's ideas to the previously exposed theories,
we are not afraid to say here that some specialists are "too young" to
understand our point of view. Missing dialogue and cooperation between specialists, as well as training
interdisciplinary, or even transdisciplinary, of the current generation of decision-makers.
Some know too much, others know little, but
no one has the patience to carry the reflections to the end.

Those who want to continue this study need to have a deep understanding of thermodynamic-economic  dictionary and Pfaff equations as nonholonomic constraints. The theory of economic equilibrium must be exploited in particular because it suggests the framework within which economic and political decision-makers can move, even if they are aware or not. In addition, the theory of black economic holes can be detailed up to numerical cases by making simulations on computers.
To these can be added the geometry of the black holes, which, we are convinced, will produce surprises, if we know how to return the results from geometry in economic context.

{\bf Authors addresses}: Prof. Emeritus Dr. Constantin Udriste, Prof. Dr. Ionel Tevy,
 University  Politehnica of Bucharest, Faculty of Applied Sciences,
Department of Mathematics-Informatics, Splaiul Independentei 313, Bucharest, 060042, Romania.

E-mail addresses: udriste@mathem.pub.ro ; tevy@mathem.pub.ro

Prof. Dr. Massimiliano Ferrara, Di.Gi.ES, University Mediterranea of Reggio Calabria,
Decisions LAB, Cittadella universitaria, seconda Torre, Loc. Feo di Vito, 89125 Reggio Calabria, Italy.

E-mail address: massimiliano.ferrara@unirc.it

Acad. Dr. Dorel Zugravescu, Institute of Geodynamics "Sabba S. Stefanescu", Dr. Gerota 19-21,  Bucharest, 020032, Romania.

E-mail address: dorezugr@geodin.ro

Associate Prof. Dr. Florin Munteanu, Institute of Geodynamics "Sabba S. Stefanescu", Dr. Gerota 19-21,  Bucharest, 020032, Romania.

E-mail address: florin@geodin.ro

\begin{thebibliography}{19}

\bibitem{[A]} J. E. Anderson, {\it The gravity model}, Working Paper 16576, http://www.nber.org/papers/w16576.

\bibitem{[0]} T. Andresen, 1994. {\it Economic Black Holes - The dynamics and consequences of accumulation}, Papers deposited by Authors, 017, Post-Keynesian Archive.

\bibitem{[01]} T. Chaney, {\it The gravity equation in international trade: an explanation},
    Department of Economics, The University of Chicago, Chicago, IL 60637.

\bibitem{[1]} D. S. Chernavski, N. I. Starkov, A. V. Shcherbakov, 2002. {\it On some problems of physical economics},
Phys. Usp. 45, 9, pp. 977-997.

\bibitem{[2]} T. I. Cre\c tu, 1996. {\it Physics} (in Romanian), Technical Editorial House, Bucharest.


\bibitem{[D]} A. Dabholkar, {\it Lectures on Quantum Black Holes}, arXiv:1208.4814v1 [hep-th] 23 Aug 2012.

\bibitem{[AD]} A. Deardorff, {\it Determinants of Bilateral Trade: Does Gravity Work
 in a Neoclassical World?}, The Regionalization of the World Economy, 1998.

\bibitem{[3]} N. Georgescu-Roegen, 1971. {\it The Entropy Law and Economic Process}, Cambridge, Mass., Harvard University Press.

\bibitem{[4]} A. Holden, 1996. {\it The New Science of Complexity}, Springer.

\bibitem{[41]} W. Isard, {\it Location theory and trade theory: short-run analysis}, Quarterly Journal of Economics, 68 (2)
(1954), 305. doi:10.2307/1884452


\bibitem{[5]} J. E. Lovelock, 1987. {\it Gaia: A New Look at Life on Earth}, Oxford University Press, Oxford New York.

\bibitem{[6]} T. M. Lenton, 1998. {\it Gaia and natural selection}, Nature, 394, pp. 439 - 447.

\bibitem{[7]} S. London, F. Tohme, 2007. {\it Thermodynamics and Economic Theory: a conceptual discussion},
Asociaci\' on Argentina de Econom\' ia Pol\' itica - XXX Reuni\' on Anual, Facultad de Ciencias Econ\' omicas - Universidad Nacional de R\' io Cuarto.

\bibitem{[8]}  C. Popescu, A. Ta\c snadi, 2009. {\it Respiritualization. Learn to be MAN} (in Romanian), ASE Editorial House, Bucharest.

\bibitem{[81]} W. Rindler, {\it Hyperbolic motion in curved space time}, Physical Review, 119, 6 (1960), 2082-2089.

\bibitem{[82]} W. Rindler, {\it Kruskal space and the uniformly accelerated frame}, American Journal of Physics, 34, 12 (1966), 1174-1178.

\bibitem{[9]} M. Ruth, 2005, {\it Insights from thermodynamics for the analysis of economic processes}, in A. Kleidon, R. Lorenz (Eds), {\it Non-equilibrium thermodynamics and the production of entropy: life, earth, and beyond}, Springer, Heidelberg, pp. 243-254.

\bibitem{[91]} J.  Tinbergen, {\it An Analysis of World Trade Flows}, in Shaping the World Economy,
edited by Jan Tinbergen. New York, NY: Twentieth Century Fund, 1962.

\bibitem{[10]} C. Udri\c ste, O. Dogaru, I. \c Tevy, 2002. {\it Extrema with Nonholonomic Constraints},
Monographs and Textbooks 4, Geometry Balkan Press.

\bibitem{[11]} C. Udri\c ste, I. \c Tevy, M. Ferrara, 2002. {\it Nonholonomic economic systems},
in C. Udri\c ste, O. Dogaru, I. Tevy, {\it Extrema with Nonholonomic Constraints},
Monographs and Textbooks 4, Geometry Balkan Press, pp. 139-150.

\bibitem{[12]} C. Udri\c ste, M. Ferrara, D. Opri\c s, 2004. {\it Economic Geometric Dynamics},
Monographs and Textbooks 6, Geometry Balkan Press.

\bibitem{13} C. Udri\c ste, 2007. {\it Thermodynamics versus Economics}, University Politehnica of Bucharest,
Scientific  Bulletin, Series A, 69, 3, pp. 89-91.

\bibitem{14} C. Udri\c ste, M. Ferrara, F. Munteanu, D. Zugr\u avescu, 2007. {\it Economics of Roegen-Ruppeiner-Weinhold type},
The International Conference of Differential Geometry and Dynamical Systems, University Politehnica of Bucharest.

\bibitem {[15]} C. Udri\c ste, M. Ferrara, 2007. {\it Multi-time optimal economic growth}, Journal of the Calcutta Mathematical Society, 3, 1, pp. 1-6.

\bibitem{[16]}  C. Udri\c ste, M. Ferrara, 2008. {\it Black hole models in economics}, Tensor,
N.S., vol. 70, 1, pp. 53-62.

\bibitem {[17]} C. Udri\c ste, M. Ferrara, 2008. {\it Multitime models of optimal growth},
WSEAS Transactions on Mathematics, 7, 1, pp. 51-55.


\bibitem {[18]} C. Udri\c ste, M. Ferrara, D. Zugr\u avescu, F. Munteanu, 2008. {\it Geobiodynamics and Roegen type economy},
Far East J. Math. Sci. (FJMS) 28, No. 3, pp. 681 - 693.

\bibitem {[UFZM]} C. Udri\c ste, M. Ferrara, D. Zugr\u avescu, F. Munteanu, 2010. {\it Nonholonomic geometry of economic systems},
ECC'10 Proceedings of the 4th European Computing Conference, Bucharest, Romania, pp. 170-177.


\bibitem{[19]} C. Udri\c ste, M. Ferrara, D. Zugr\u avescu, F. Munteanu, 2012. {\it Controllability
of a nonholonomic macroeconomic system}, J. Optim. Theory Appl., 154, 3, pp. 1036-1054.


\bibitem{[20]} C. Udri\c ste, {\it Optimal control on nonholonomic black holes},
Journal of Computational Methods in Sciences and Engineering, 13, 1-2 (2013) pp. 271-278.

\bibitem{[21]} www.complexity.ro/ONCE-CS-RoadMap-V22.pdf; http://www.complexity.ro/ONCE-CS-RoadMap-V22.pdf, 2006.



\end{thebibliography}
\end{document}